\documentclass[11pt]{article}
\usepackage{amsmath,amssymb,amsfonts,latexsym,graphicx,amsthm}
\usepackage{fullpage,color}
\usepackage{url,hyperref}
\usepackage{comment}
\usepackage[linesnumbered,boxed,ruled,vlined]{algorithm2e}
\usepackage{framed}
\usepackage[shortlabels]{enumitem}

\usepackage[margin=1in]{geometry}

\newtheorem{theorem}{Theorem}[section]
\newtheorem{lemma}{Lemma}[section]
\newtheorem{definition}{Definition}[section]

\newtheorem{claim}{Claim}[lemma]

        {\medskip}

\newcommand{\wts}{\omega}

\begin{document}

\title{Gomory-Hu Trees in Quadratic Time}
\author{Tianyi Zhang \thanks{Tel Aviv University, \href{}{tianyiz21@tauex.tau.ac.il}}}
\date{}

\maketitle

\begin{abstract}
	Gomory-Hu tree [Gomory and Hu, 1961] is a succinct representation of pairwise minimum cuts in an undirected graph. When the input graph has general edge weights, classic algorithms need at least cubic running time to compute a Gomory-Hu tree. Very recently, the authors of [AKL+, arXiv v1, 2021] have improved the running time to $\tilde{O}(n^{2.875})$ which breaks the cubic barrier for the first time. In this paper, we refine their approach and improve the running time to $\tilde{O}(n^2)$. This quadratic upper bound is also obtained independently in an updated version by the same group of authors [AKL+, arXiv v2, 2021].
\end{abstract}

\thispagestyle{empty}
\clearpage
\setcounter{page}{1}
\pagestyle{plain}

\section{Introduction}
It is a famous result by Gomory and Hu \cite{gomory1961multi} that any undirected graph can be compressed into a single tree while all pairwise minimum cuts are preserved exactly. More specifically, given any weighted undirected graph $G = (V, E, \wts)$ on $n$ vertices and $m$ edges, there exists an edge weighted spanning tree $T$ on the same set of vertices $V$, such that: for any pair of vertices $s, t\in V$, the minimum $(s, t)$-cut in $T$ is also a minimum $(s, t)$-cut in $G$, and their cut values are equal. Such trees are called Gomory-Hu trees. In the original paper \cite{gomory1961multi}, Gomory and Hu showed an algorithm that reduces the task of constructing a Gomory-Hu tree to $n-1$ max-flow instances. Gusfield \cite{gusfield1990very} modified the original algorithm Gomory and Hu so that no graph contractions are needed when applying max-flow subroutines. So far, in weighted graphs, faster algorithms for building Gomory-Hu trees were only byproducts of faster max-flow algorithms. In the recent decade, there has been a sequence of improvements on max-flows using the interior point method \cite{lee2014path,madry2016computing,liu2020faster,kathuria2020unit,brand2021minimum}, and the currently best running time is $\tilde{O}(m + n^{1.5})$\footnote{$\tilde{O}$ suppresses poly-logarithimic factors.} by \cite{brand2021minimum}, hence using the classical reductions, computing a Gomory-Hu tree requires $\tilde{O}(mn + n^{2.5})$ time, which is still cubic in dense graphs.

In the special case where $G$ is a simple graph ($\wts$ only assigns unit weights and no parallel edges exist), several improvements have been made over the years. The authors of \cite{hariharan2007mn} designed a $\tilde{O}(mn)$ time algorithm for Gomory-Hu trees using a tree packing approach based on \cite{gabow1995matroid,edmonds2003submodular}. In a recent line of works~\cite{abboud2021friendly,abboud2021apmf,li2021nearly,zhang2021faster,abboud2020subcubic,abboud2020new},  the running time has been improved significantly to $\widehat{O}(m + n^{1.9})$\footnote{$\widehat{O}$ hides sub-polynomial factors.} which is sub-quadratic.

For general edge-weighted graphs, the longstanding cubic barrier of \cite{gomory1961multi} has been surpassed by a very recent online preprint~\cite{abboud2021gomory} (version 1). More specifically, the authors proposed an algorithm that computes a Gomory-Hu tree of $G$ in $\tilde{O}(n^{3-1/8})$ time.

\subsection{Our result}
In this paper, we refine the approach of \cite{abboud2021gomory} and achieve a faster running time.
\begin{theorem}\label{quad}
	Let $G = (V, E, \wts)$ be an undirected graph on $n$ vertices and $m$ edges with positive integer weights, there is a randomized algorithm that computes a Gomory-Hu tree of $G$ in $\tilde{O}(n^2)$ time; when $G$ is unweighted, the running time becomes $\widehat{O}(m + n^{1.5})$.
\end{theorem}

\noindent\textbf{Independent work.} The same result has also been achieved independently in an updated version of \cite{abboud2021gomory} (version 2).

\subsection{Technical overview}
\textbf{Notations.} Denote by $G = (V, E, \wts)$ be an undirected graph whose edge weights are positive integers. For any pair of vertices $a, b\in V$, let $\lambda(a, b)$ be the value of the minimum $(a, b)$-cut. For any vertex subset $U\subseteq V$, the $U$-Steiner minimum cut is denoted by $\lambda(U) = \min_{a, b\in U}\lambda(a, b)$. For any subset $A\subseteq V$, define $\delta(A)$ to be the value of the cut $(A, V\setminus A)$ in $G$.

\begin{definition}[The single-source terminal min-cuts verification (SSTCV) problem]
	The input is a graph $G = (V, E)$, a terminal set $U\subseteq V$ and a source terminal $s\in U$ with the guarantee that for all $t\in U\setminus \{s\}$, $\lambda(U)\leq \lambda(s, t)\leq 1.1\lambda(U)$, and estimates $\{\mu_t \}_{t\in U\setminus \{s\}}$ such that $\mu_t\geq \lambda(s, t)$. The goal is to determine for each $t\in U\setminus \{s\}$ whether $\mu_t = \lambda(s, t)$.
\end{definition}

\begin{lemma}[\cite{li2021nearly}]
	There is a randomized algorithm that computes a Gomory-Hu tree of an input graph by making calls to max-flow and SSTCV on graphs with a total of $\tilde{O}$ vertices and $\tilde{O}(m)$ edges, and runs for $\tilde{O}(m)$ time outside of these calls.
\end{lemma}

A tree $T$ on $V$ is called a $U$-Steiner tree if it spans $U$; when
$U$is clear from the context, we write Steiner instead of $U$-Steiner. We stress that $T$ need not be a subgraph of $G$. Let $(A, V\setminus A)$ be a cut and $C\subseteq E$ be the edge set crossing this cut. We say that a tree $T$ $k$-respects a cut $(A, V\setminus A)$ (and vice versa) if $|E(T)\cap C|\leq k$.

\begin{definition}[Guide trees~\cite{abboud2021gomory}] For a graph $G$ and set of terminals $U\subseteq V$with a source $s\in U$, a collection of Steiner trees $T_1, \cdots, T_h$ is called a $k$-respecting set of guide trees, or in short guide trees, if for every $t\subseteq U\setminus \{s\}$, at least one tree $T_i$ is $k$-respecting at least one minimum $(s, t)$-cut in
$G$.
\end{definition}

As a preparatory step, we first need an algorithm from~\cite{abboud2021gomory} that computes guide trees.
\begin{lemma}[\cite{abboud2021gomory}]\label{guide}
	There is a randomized algorithm that, given a graph $G = (V, E, \wts)$, a terminal set $U\subseteq V$ and a source terminal $s\in U$, with the guarantee that for $t\in U\setminus \{s\}$, $\lambda(U)\leq \lambda(s, t)\leq 1.1\lambda(U)$, computes a $4$-respecting set of $O(\log n)$ guide trees. The algorithm takes $\tilde{O}(n^2)$ time in weighted graphs, and $\widehat{O}(m)$ in unweighted graphs.
\end{lemma}

According to~\cite{abboud2021gomory}, it suffices to solve the problem of \emph{single-source min-cuts given a guide-tree}. Our main technical contribution is the following statement.
\begin{lemma}\label{tree}
	Let $G = (V, E, \wts)$ be an undirected weighted graph on $n$ vertices and $m$ edges, with terminals $U\subseteq V$ containing a given source vertex $s\in U$, and let $T$ be a $U$-Steiner tree. Let $k\geq 2$ be a constant independent of $n, m$. Then, there is a randomized algorithm that computes for each terminal $t\in U$ a value $\mu(s, t)$ such that: $$\lambda(s, t)\leq \mu(s, t)\leq \lambda_{T, k}(s, t)$$ where $\lambda_{T, k}(s, t)$ is the value of the minimum $(s, t)$-cut that $k$-respects tree $T$, while $\lambda(s, t)$ is the value of minimum $(s, t)$-cut in $G$.
	
	The algorithm is a near-linear time reduction to max-flow instances with a total number of $\tilde{O}(m)$ edges and $\tilde{O}(n)$ edges. Using the currently fastest max-flow algorithm~\cite{van2021minimum}, the running time becomes $\tilde{O}(m + n^{1.5})$.
\end{lemma}
For comparison, the algorithm in~\cite{abboud2021gomory} (version 1) runs in time $\tilde{O}(n^{3-1/2^{k-1}})$. Combining the above lemma with Lemma~\ref{guide}, it immediately proves Theorem~\ref{quad}.

\section{Faster single-source minimum cuts}
In this section we prove Lemma~\ref{tree}. Our algorithm mostly follows the recursive framework of~\cite{abboud2021gomory}, but with some generalizations. The algorithm will be named $\mathsf{TreeMincuts}(G, T, k)$ which accepts parameters $(G, T, k)$ where $T$ is a Steiner tree. Throughout the recursion, the algorithm maintains for each $t\in V(T)\setminus \{s\}$ a candidate cut value $\mu(s, t)$ which is initialized to be $\infty$. 

Other than $\mathsf{LeafMincuts}$, we also need a helper subroutine $\mathsf{LeafMincuts}(G, T, k)$ that deals with the special case where $s$ is a leaf of $T$ connected via edge $(s, p)$, and for each $t\in V(T)\setminus \{s\}$, we only need to ensure that $\mu(s, t)\leq \eta_{T, k}(s, t)$, where $\eta_{T, k}(s, t)$ is the value of the minimum $k$-respecting $(s, t)$-cut which crosses edge $(s, p)$.

\subsection{Description of $\mathsf{TreeMincuts}$}
If $k = 0$, then the algorithm does nothing. For the base case where $|T| = O(1)$, then compute each minimum $(s, t)$-cut in $G$ by applying standard max-flow algorithms, and then terminate from here. For the rest let us assume $|T|$ is at least a large constant. The algorithm consists of several phases below.\\

\noindent \textbf{Preparation.} Let $c\in V(T)$ be a centroid; that is $T\setminus \{c\}$ is a forest of spanning trees of size at most $2|V(T)|/3$. Let $T_1, \cdots, T_l$ be subtrees of $T\setminus \{c\}$ not containing $s$. When $s\neq c$, consider the tree path of $T$ from $s$ to $c$ denoted by $\langle (s =) r_0, r_1, \cdots, r_h (= c) \rangle$. Consider the forest $T\setminus \{r_0, r_1, \cdots, r_h\}$. Let $F\subseteq T$ be the forest connected to vertex $s$ in $T\setminus \{r_1\}$, and let $T_0\subseteq T$ be the subtree $T\setminus (\{s\}\cup V(F)\cup\bigcup_{i=1}^l V(T_i))$.\\

\noindent \textbf{Subtree pruning.} For each subtree $T_i, 1\leq i\leq l$, and the forest $F$, independently prune it from $T$ with probability $1/2$. Let $T^{(1)}$ be the randomly pruned tree, and then recursively call $\mathsf{TreeMincuts}(G, T^{(1)}, k-1)$. Repeat this procedure for $O(\log n)$ times.

After that, define $T^{(2)}\subseteq T$ to be the subtree on $\{s\}\cup V(F)$. Recursively call the routine $\mathsf{TreeMincuts}(G, T^{(2)}, k-1)$.\\

\noindent \textbf{Isolating cuts.} Compute minimum isolating cuts on $G$ with terminal sets $\{s, c\}$, all of $V(T_i)$ for any $0\leq i\leq l$, and $V(F)$. For each $0\leq i\leq l$, let $W_i\supseteq V(T_i)$ be the side of the isolating cut containing $V(T_i)$. For each $0\leq i\leq l$, contract all vertices $V\setminus W_i$ into a single node, so $T$ has become a subtree on $\{c\}\cup V(T_i)$. Relabel $c$ as $s$, and call the new tree $T^{(3)}_i$, and the contracted graph $G_i$. Recursively call $\mathsf{TreeMincuts}(G_i, T^{(3)}_i, k)$.

Similarly, let $S\supseteq V(F)$ be the side of the isolating cut containing $V(F)$. Contract all vertices from $V\setminus S$ into a single node, so $T$ has become a subtree on $\{s\}\cup V(F)$. Call this new tree $T^{(3)}_0$, and the contracted graph $H$. Recursively call $\mathsf{TreeMincuts}(H, T^{(3)}_0, k)$.

Finally, delete vertices in $V(F)\cup V(T_0)$ from $T$, and attach $s$ directly to $c$, which creates a new tree $T^{(4)}$. Recursively call $\mathsf{LeafMincuts}(G, T^{(4)}, k)$ and $\mathsf{TreeMincuts}(G, T^{(4)}, k-1)$.

\subsection{Description of $\mathsf{LeafMincuts}$}
The algorithm is mostly similar to $\mathsf{TreeMincuts}$, but with some subtle differences. When $k = 0$, then the algorithm does nothing. For the base case where $|T| = O(1)$, then compute each minimum $(s, t)$-cut in $G$ by applying standard max-flow algorithms, and then terminate from here. For the rest let us assume $|T|$ is at least a large constant. The algorithm consists of several phases below.\\

\noindent \textbf{Preparation.} Let $c\in V(T)$ be a centroid; that is $T\setminus \{c\}$ is a forest of spanning trees of size at most $2|V(T)|/3$. Let $T_0, T_1, \cdots, T_l$ be subtrees of $T\setminus \{c\}$ such that $s\in V(T_0)$. Let $(s, p)\in E(T)$ be the unique edge incident on $s$.\\

\noindent \textbf{Subtree pruning.} For each subtree $T_i, 1\leq i\leq l$, independently prune it from $T$ with probability $1/2$. Let $T^{(1)}$ be the randomly pruned tree, and then recursively call $\mathsf{LeafMincuts}(G, T^{(1)}, k-1)$. Repeat this procedure for $O(\log n)$ times.

In addition, remove $V(T_0)$ from $T$ entirely, and reconnect $s$ to $T$ at $c$. Call the new tree $T^{(2)}$. Recursively call $\mathsf{TreeMincuts}(G, T^{(2)}, k-1)$.\\

\noindent \textbf{Isolating cuts.} Compute minimum isolating cuts on $G$ with terminal sets $\{s, c\}$, all of $V(T_i)$ for any $1\leq i\leq l$, and $V(T_0)\setminus \{s\}$. For each $1\leq i\leq l$, let $W_i\supseteq V(T_i)$ be the side of the isolating cut containing $V(T_i)$, and let $W_0\supseteq V(T_0)\setminus \{s\}$ be the side of the isolating cut containing $V(T_0)\setminus \{s\}$. 

Contract all vertices $V\setminus W_0$ into a single node. By doing this, all vertices in $\bigcup_{i=1}^lV(T_i)$ have been merged with $c$, and $T$ has become a tree on $\{c\}\cup V(T_0)$. Call the new tree $T^{(3)}$, the contracted graph $G_0$. Recursively call $\mathsf{LeafMincuts}(G_0, T^{(3)}, k)$.

The key difference from $\mathsf{TreeMincuts}$ is how it deals with vertices in $\bigcup_{i=1}^l V(T_i)$. As we will see later, it is also very important that this step goes after the subtree pruning phase. 

First, apply standard max-flow algorithms to compute the minimum $(s, \{c\}\cup\bigcup_{i=1}^lV(T_i))$-cut, and update the cut value to all $\mu(s, t)$, $t\in \{c\}\cup\bigcup_{i=1}^lV(T_i)$. Take the vertex $z\in \bigcup_{i=1}^lV(T_i)$ such that the value of $\mu(s, z)$ is \textbf{maximized} (ties are broken arbitrarily), under the current view of cut values $\mu(\cdot, \cdot)$. Without loss of generality, assume $z\in V(T_1)$. Then, contract all $W_i, \forall i\neq 1$ into a single node and merge it with $c$. Thus, $T$ has become a tree on vertices $\{s, c\}\cup V(T_1)$, and call this new tree $T^{(4)}$, the contracted $G_1$. Then, recursively call $\mathsf{LeafMincuts}(G_1, T^{(4)}, k)$. We want to emphasize that due to efficiency concerns, we cannot recursively call $\mathsf{LeafMincuts}$ on other trees $T_j, 1<j\leq l$, since otherwise these instances would have significant overlaps.

\subsection{Proof of correctness}
We will prove this correctness inductively on the recursion depth, and the induction switches between $\mathsf{TreeMincuts}$ and $\mathsf{LeafMincuts}$. All correctness will hold with high probability. Below we state what correctness means for subroutines $\mathsf{TreeMincuts}$ and $\mathsf{LeafMincuts}$, respectively.
\begin{itemize}[leftmargin=*]
	\item $\mathsf{TreeMincuts}(G, T, k)$ succeeds if it returns values $\mu(s, t)$ such that $\lambda(s, t)\leq \mu(s, t)\leq \lambda_{T, k}(s, t)$, where $\lambda_{T, k}(s, t)$ is the value of the minimum $(s, t)$-cut that $k$-respects tree $T$.
	\item $\mathsf{LeafMincuts}(G, T, k)$ succeeds if it returns values $\mu(s, t)$ such that $\lambda(s, t)\leq \mu(s, t)\leq \eta_{T, k}(s, t)$ where $\eta_{T, k}(s, t)$ is the value of the minimum $k$-respecting $(s, t)$-cut which crosses edge $(s, p)$; here we assume $s$ is a leaf of $T$ connected via edge $(s, p)$.
\end{itemize}

Our inductive proof of correctness is a combination of the following two statements.
\begin{lemma}
	Suppose that $\mathsf{TreeMincuts}$ and $\mathsf{LeafMincuts}$ succeeds on smaller inputs. Then, we claim $\mathsf{TreeMincuts}(G, T, k)$ returns for every $t\in V(T)$ a value $\mu(s, t)$ such that $\lambda(s, t)\leq \mu(s, t)\leq \lambda_{T, k}(s, t)$.
\end{lemma}
\begin{proof}
	It is easy to argue that lower bounds always hold, namely $\mu(s, t)\geq \lambda(s, t)$. So for the rest we only focus on the upper bounds $\mu(s, t)\leq \lambda_{T, k}(s, t)$. Let $t\in V(T)\setminus \{s, c\}$ be an arbitrary vertex. Let $(A, V\setminus A)$ be the minimum $(s, t)$-cut which $k$-respects $T$, with $t\in A, s\notin A$.
	
	\begin{framed}
	\noindent First let us study vertices $t\in V(F)$. For this let us study two cases.
	\begin{itemize}[leftmargin=*]
		\item Suppose $(A, V\setminus A)$ only cuts tree edges incident on $F$. Then we claim that after calling $\mathsf{TreeMincuts}(H, T^{(3)}_0, k)$, $\mu(s, t)$ becomes at most $\lambda_{T, k}(s, t)$. It suffices to show that there exists a subset $B\subseteq S$ with $t\in B, s\notin B$ such that $\delta(B) = \delta(A)$ in graph $G$. By sub-modularity of cuts, we have:
		$$\delta(A) + \delta(S)\geq \delta(A\cap S) + \delta(A\cup S)$$
		However, on the one hand, since $(A, V\setminus A)$ does not cut any tree edges not incident on $F$, $A\cup S$ does not contain any vertices in $V(T)\setminus V(F)$. Hence, by minimality of $(S, V\setminus S)$ we know that $\delta(S)\leq \delta(A\cup S)$. On the other hand, since $V(F)\subseteq S$, we know that $(A\cup S, V\setminus (A\cup S))$ cuts the same number of edges on $T$ as $(A, V\setminus A)$. So by minimality of $(A, V\setminus A)$, we have $\delta(A)\leq \delta(A\cap S)$. Summing both ends, we have $\delta(A) + \delta(S) =\delta(A\cap S) + \delta(A\cup S)$, and so both equalities hold. In particular, $\delta(A) = \delta(A\cap S)$. Taking $B = A\cap S$ suffices.
		
		\item Suppose $(A, V\setminus A)$ cuts tree edges not incident on $F$, and so it cuts at most $k-1$ tree edges incident on $F$. Hence, the recursive call on $\mathsf{TreeMincuts}(G, T^{(2)}, k-1)$ would update $\mu(s, t)$ desirably.
	\end{itemize}
	\end{framed}
	
	\begin{framed}
	\noindent Now suppose $t \in V(T_0)$. There are several cases below.
	\begin{itemize}[leftmargin=*]
		\item Suppose $(A, V\setminus A)$ only cuts tree edges incident on $T_0$. Then we claim that after calling $\mathsf{TreeMincuts}(G_0, T_0^{(3)}, k)$, $\mu(s, t)$ becomes at most $\lambda_{T, k}(s, t)$. It suffices to show that there exists a subset $B\subseteq W_0$ with $t\in B, s\notin B$ such that $\delta(B) = \delta(A)$ in graph $G$. By sub-modularity of cuts, we have:
		$$\delta(A) + \delta(W_0)\geq \delta(A\cap W_0) + \delta(A\cup W_0)$$
		However, on the one hand, since $(A, V\setminus A)$ does not cut any tree edges not incident on $T_0$, $A\cup W_0$ does not contain any vertices in $V(T)\setminus V(T_0)$. Hence, by minimality of $(W_0, V\setminus W_0)$ we know that $\delta(W_0)\leq \delta(A\cup W_0)$. On the other hand, since $V(T_0)\subseteq W_0$, we know that $(A\cap W_0, V\setminus (A\cap W_0))$ cuts the same number of edges on $T$ as $(A, V\setminus A)$. So by minimality of $(A, V\setminus A)$, we have $\delta(A)\leq \delta(A\cap W_0)$. Summing both ends, we have $\delta(A) + \delta(W_0) =\delta(A\cap W_0) + \delta(A\cup W_0)$, and so both equalities hold. In particular, $\delta(A) = \delta(A\cap W_0)$. Taking $B = A\cap W_0$ suffices.
		
		\item Suppose $(A, V\setminus A)$ cuts some edges incident on forest $F$ or other trees $T_i, 1\leq i\leq l$, then with at least constant probability over the choice of $T^{(1)}$, $T_0$ stays in $T^{(1)}$ while $F$ or $T_i$ gets pruned. Hence, the recursive call on $\mathsf{TreeMincuts}(G, T^{(1)}, k-1)$ updates $\mu(s, t)$ correctly.
	\end{itemize}
	\end{framed}
	
	\begin{framed}
		\noindent Let us consider vertices $t\in V(T_i)$ for some $1\leq i\leq l$.
		\begin{itemize}[leftmargin=*]
			\item Suppose $(A, V\setminus A)$ only cuts tree edges incident on $T_i$. Then we claim that after calling $\mathsf{TreeMincuts}(G_i, T^{(3)}_i, k)$, $\mu(s, t)$ becomes at most $\lambda_{T, k}(s, t)$. It suffices to show that there exists a subset $B\subseteq W_i$ with $t\in W_i, s\notin W_i$ such that $\delta(B) = \delta(A)$ in graph $G$. By sub-modularity of cuts, we have:
			$$\delta(A) + \delta(W_i)\geq \delta(A\cap W_i) + \delta(A\cup W_i)$$
			However, on the one hand, since $(A, V\setminus A)$ does not cut any tree edges not incident on $T_i$, $A\cup W_i$ does not contain any vertices in $V(T)\setminus V(T_i)$. Hence, by minimality of $(S, V\setminus W_i)$ we know that $\delta(W_i)\leq \delta(A\cup W_i)$. On the other hand, since $V(T_i)\subseteq W_i$, we know that $(A\cup W_i, V\setminus (A\cup W_i))$ cuts the same number of edges on $T$ as $(A, V\setminus A)$. So by minimality of $(A, V\setminus A)$, we have $\delta(A)\leq \delta(A\cap W_i)$. Summing both ends, we have $\delta(A) + \delta(W_i) =\delta(A\cap W_i) + \delta(A\cup W_i)$, and so both equalities hold. In particular, $\delta(A) = \delta(A\cap W_i)$. Taking $B = A\cap W_i$ suffices.
			
			\item Suppose $(A, V\setminus A)$ cuts some edges incident on forest $F$ or other trees $T_j, j\neq i$, then with at least constant probability over the choice of $T^{(1)}$, $T_i$ stays in $T^{(1)}$ while $F$ or $T_j$ gets pruned. Hence, the recursive call on $\mathsf{TreeMincuts}(G, T^{(1)}, k-1)$ updates $\mu(s, t)$ correctly.
			
			\item Suppose $(A, V\setminus A)$ does not cut any edges incident on forest $F$ nor other trees $T_j, j\neq i$, but some edges in $T_0$. If $c\notin A$, then $(A, V\setminus A)$ does not separate $s, c$ in $T^{(4)}$, so it has at most $k-1$ cut edges in $T^{(4)}$. Hence this case is covered in recursive call $\mathsf{TreeMincuts}(G, T^{(4)}, k-1)$.
			
			If $c\in A$, so $(A, V\setminus A)$ cuts $(s, c)$ on $T^{(4)}$. Therefore, this case is covered in recursive call $\mathsf{LeafMincuts}(G, T^{(4)}, k)$.
		\end{itemize}
	\end{framed}
\end{proof}

\begin{lemma}
	Suppose that $\mathsf{TreeMincuts}$ and $\mathsf{LeafMincuts}$ succeeds on smaller inputs. Then, we claim $\mathsf{LeafMincuts}(G, T, k)$ returns for every $t\in V(T)$ a value $\mu(s, t)$ such that $\lambda(s, t)\leq \mu(s, t)\leq \eta_{T, k}(s, t)$.
\end{lemma}
\begin{proof}
		It is easy to argue that lower bounds always hold, namely $\mu(s, t)\geq \lambda(s, t)$. So for the rest we only focus on the upper bounds $\mu(s, t)\leq \lambda_{T, k}(s, t)$. Let $t\in V(T)\setminus \{s, c\}$ be an arbitrary vertex. Let $(A, V\setminus A)$ be the minimum $(s, t)$-cut which $k$-respects $T$, with $p, t\in A, s\notin A$.
		
	\begin{framed}
		\noindent Suppose $t \in V(T_0)$. There are several cases below.
		\begin{itemize}[leftmargin=*]
			\item Suppose $(A, V\setminus A)$ only cuts tree edges incident on $T_0$. Then we claim that after calling $\mathsf{LeafMincuts}(G_0, T^{(3)}, k)$, $\mu(s, t)$ becomes at most $\eta_{T, k}(s, t)$. It suffices to show that there exists a subset $B\subseteq W_0$ with $t\in B, s\notin B$ such that $\delta(B) = \delta(A)$ in graph $G$. By sub-modularity of cuts, we have:
			$$\delta(A) + \delta(W_0)\geq \delta(A\cap W_0) + \delta(A\cup W_0)$$
			However, on the one hand, since $(A, V\setminus A)$ does not cut any tree edges not incident on $T_0$, $A\cup W_0$ does not contain any vertices in $V(T)\setminus V(T_0)$. Hence, by minimality of $(W_0, V\setminus W_0)$ we know that $\delta(W_0)\leq \delta(A\cup W_0)$. On the other hand, since $V(T_0)\subseteq W_0$, we know that $(A\cap W_0, V\setminus (A\cap W_0))$ cuts the same number of edges on $T$ as $(A, V\setminus A)$. So by minimality of $(A, V\setminus A)$, we have $\delta(A)\leq \delta(A\cap W_0)$. Summing both ends, we have $\delta(A) + \delta(W_0) =\delta(A\cap W_0) + \delta(A\cup W_0)$, and so both equalities hold. In particular, $\delta(A) = \delta(A\cap W_0)$. Taking $B = A\cap W_0$ suffices.
			
			\item Suppose $(A, V\setminus A)$ cuts some edges incident some trees $T_i, 1\leq i\leq l$, then with at least constant probability over the choice of $T^{(1)}$, $T_0$ stays in $T^{(1)}$ while $F$ or $T_i$ gets pruned. Hence, the recursive call on $\mathsf{LeafMincuts}(G, T^{(1)}, k-1)$ updates $\mu(s, t)$ correctly.
		\end{itemize}
	\end{framed}
	
		Let us consider vertices $t\in V(T_i)$ for some $1\leq i\leq l$.
		\begin{itemize}[leftmargin=*]
			\item Suppose $(A, V\setminus A)$ cuts some edges incident on other trees $T_j, j\neq i$, then with at least constant probability over the choice of $T^{(1)}$, $T_i$ stays in $T^{(1)}$ while $F$ or $T_j$ gets pruned. Hence, the recursive call on $\mathsf{LeafMincuts}(G, T^{(1)}, k-1)$ updates $\mu(s, t)$ correctly.
			
			\item Suppose $(A, V\setminus A)$ does not cut any edges incident on tree $T_j, j\neq i, 1\leq j\leq l$, plus the condition that either $c\notin A$, or $(A, V\setminus A)$ cuts some edges incident on $T_0$ other than $(s, p)$.
			
			In this case, if $c\notin A$, then $(A, V\setminus A)$ does not separate $s, c$ in $T^{(2)}$. As $(A, V\setminus A)$ cuts $(s, p)$ in $T$, so it has at most $k-1$ cut edges in $T^{(2)}$. Hence this case is covered in recursive call $\mathsf{TreeMincuts}(G, T^{(2)}, k-1)$.
			
			$(A, V\setminus A)$ cuts some edges incident on $T_0$ other than $(s, p)$, then as $(A, V\setminus A)$ cuts $(s, p)$ in $T$, it has at most $k-2$ cut edges in $T^{(2)}$. Hence this case is covered in recursive call $\mathsf{TreeMincuts}(G, T^{(2)}, k-1)$.
			
			\item The most important case is that $(A, V\setminus A)$ does not cut any edges incident on tree $T_j, j\neq i, 0\leq j\leq l$, plus the condition that $c\in A$.
			
			If $(A, V\setminus A)$ also avoid any edges incident on $T_i$, then $A$ contains the entire $\{c\}\cup \bigcup_{j=1}^l V(T_l)$. Then this case is covered by the max-flow computation that computes the minimum $(s, \{c\}\cup\bigcup_{i=1}^lV(T_i))$-cut. For the rest, let us assume $(A, V\setminus A)$ cuts at least one tree edge incident on $T_i$.
			
			\begin{claim}
				In this case, it must be $i = 1$.
			\end{claim} 
			\begin{proof}[Proof of claim]
				Suppose otherwise that $2\leq i\leq l$. Since $(A, V\setminus A)$ avoid any tree edges incident on $T_j, j\neq i, 0\leq j\leq l$, plus that $c\in A$, we know that $\bigcup_{j=1, j\neq i}^l V(T_j)\subseteq A$. So, in particular we have $z\in A$; recall that $z\in V(T_1)$ is a maximizer of $\mu(s, \cdot)$ among vertices in $\bigcup_{j=1}^lV(T_j)$ after the subtree pruning phase.
			
				Consider the moment when the algorithm was just about to find the maximizer. We can assume by then $\mu(s, t) > \delta(A)$. Since $z\in A$, $(A, V\setminus A)$ is an $(s, z)$-cut that $k$-respects $T$, cutting tree edge $(s, p)$, and most importantly it cuts an edge incident on tree $T_i$ different from $T_1$. Then, during subtree pruning, $\mu(s, z)$ has already become $\leq \delta(A)$ after $O(\log n)$ trials of $\mathsf{LeafMincuts}(G, T^{(1)}, k-1)$. Hence, by the time when looking for maximizers, we know $\mu(s, t) > \delta(A) \geq \mu(s, z)$, which contradicts the maximality of $z$.
			\end{proof}
			
			Next we prove that after calling $\mathsf{TreeMincuts}(G_1, T^{(4)}, k)$, $\mu(s, t)$ becomes at most $\delta(A)$. Define $W = \{c\}\cup\bigcup_{i=0, i\neq 1}^l W_i$. It suffices to show that there exists a subset $B\supseteq W$ with $t\in B, s\notin B$ such that $\delta(B) = \delta(A)$ in graph $G$. 
			
			Consider any $W_j, 0\leq j\leq l, j\neq 1$ such that $W_j\setminus A\neq \emptyset$. By sub-modularity of cuts, for each $0\leq i\leq l, j\neq 1$ we have:
			$$\delta(A) + \delta(W_j)\geq \delta(A\cap W_j) + \delta(A\cup W_j)$$
			
			On the one hand, as $(A, V\setminus A)$ avoids all tree edges incident on $\bigcup_{i=0,i\neq 1}^l V(T_i)$ except for $(s, p)$, $A\cap W_j$ still contain any vertices in $V(T_i)$ ($V(T_0)\setminus \{s\}$ when $i=0$). Hence, by minimality of $(W_i, V\setminus W_i)$ we know that $\delta(W_i)\leq \delta(A\cap W_i)$. 
			
			On the other hand, since $W_i$ is an isolating set, we know that $(A\cup W_i, V\setminus (A\cup W_i))$ cuts the same number of edges on $T_i$ as $(A, V\setminus A)$. So by minimality of $(A, V\setminus A)$, we have $\delta(A)\leq \delta(A\cup W_i)$. Summing both ends, we have $\delta(A) + \delta(W_i) =\delta(A\cap W_i) + \delta(A\cup W_i)$, and so both equalities hold. In particular, $\delta(A) = \delta(A\cup W_i)$. So we can replace $A$ with $A\cup W_i$. Iterating over all $i$ finishes the proof.
		\end{itemize}
\end{proof}

\subsection{Running time analysis}
We first bound the total number of vertices of all instances for recursive calls $\mathsf{TreeMincuts}(G, T, k)$ and $\mathsf{LeafMincuts}(G, T, k)$. Denote this total amount by a recursive function $f_1(n, t, k)$ and $f_2(n, t, k)$ respectively, where $t = |V(T)|$ and $n = |V|$.

In algorithm $\mathsf{TreeMincuts}$, for each $0\leq i\leq l$, let $n_i$ be the number of vertices in $G_i$, and $t_i = |V(T_i^{(3)})|$. So $\sum_{i=0}^l (n_i-1) \leq n-1$, $\sum_{i=0}^l (t_i-1)\leq t-1$, $t_i\leq 2t/3$. According to the algorithm description, we have the recursive relation:
$$f_1(n, t, k)\leq O(\log n)\cdot f_1(n, t, k-1) +f_2(n, t, k) + \sum_{i=0}^l f_1(n_i, t_i, k)$$

In algorithm $\mathsf{LeafMincuts}$, let $n_0$ be the number of vertices in $G_0$, and $n_1$ be the number of vertices in $G_1$, and $t_0 = |V(T^{(3)})|, t_1 = |T^{(4)}|$. So $(n_0-1) + (n_1-1)\leq n-1$, and $(t_0-1) + (t_1-1)\leq t$, $t_0, t_1\leq 2t/3$. According to the algorithm description, we have the recursive relation:
$$f_2(n, t, k)\leq O(\log n)\cdot f_2(n, t, k-1) + f_1(n, t, k-1) + f_2(n_0, t_0, k) + f_2(n_1, t_1, k)$$

Using induction we can prove that 
$$\begin{aligned}
f_1(n, t, k)&\leq (n-1)\cdot\log^{3k}n\cdot \log t\\
f_2(n, t, k)&\leq (n-1)\cdot\log^{3k-2}n\cdot \log t
\end{aligned}$$

Next we try to upper bound the total number of edges \textbf{not incident on any contracted nodes} of all instances for recursive calls $\mathsf{TreeMincuts}(G, T, k)$ and $\mathsf{LeafMincuts}(G, T, k)$. Denote this total amount by a recursive function $g_1(n, t, k)$ and $g_2(n, t, k)$ respectively, where $t = |V(T)|$ and $m = |E|$. Since each recursive instance contains at most $O(\log t)$ contracted vertices, the total number of edges can be bounded by $g_1(m, t, k) + f_1(n, t, k)\cdot O(\log t)$ and $g_2(m, t, k) + f_2(n, t, k)\cdot O(\log t)$, respectively.

In algorithm $\mathsf{TreeMincuts}$, for each $0\leq i\leq l$, let $m_i$ be the number of edges in $G_i$ not incident on contracted vertices. So $\sum_{i=0}^l m_i\leq m$. According to the algorithm description, we have the recursive relation:
$$g_1(m, t, k)\leq O(\log n)\cdot g_1(m, t, k-1) +g_2(m, t, k) + \sum_{i=0}^l g_1(m_i, t_i, k)$$

In algorithm $\mathsf{LeafMincuts}$, let $m_0 m_1$ be the number of vertices in $G_0, G_1$ not incident on contracted vertices. So $m_0 + m_1\leq m$. According to the algorithm description, we have the recursive relation:
$$g_2(m, t, k)\leq O(\log n)\cdot g_2(m, t, k-1) + g_1(m, t, k-1) + g_2(m_0, t_0, k) + g_2(m_1, t_1, k)$$

Using induction we can prove that 
$$\begin{aligned}
g_1(m, t, k)&\leq m\cdot\log^{3k}n\cdot \log t\\
g_2(m, t, k)&\leq m\cdot\log^{3k-2}n\cdot \log t
\end{aligned}$$
When $k$ is a constant, this proves Lemma~\ref{tree}.

\section*{Acknowledgement}
The author would like to thank helpful discussions with Shiri Chechik.

\vspace{5mm}
\bibliographystyle{alpha}
\bibliography{ref}

\end{document}